\documentclass[12pt]{article}
\usepackage{geometry}

\usepackage{svg}
\usepackage{amsmath, amssymb, graphicx,fullpage,color,mathtools,amsthm,xcolor}
\usepackage{caption, subcaption}
\usepackage{algorithm, algorithmic}
\usepackage{cite}
\usepackage{subcaption}
\usepackage{array}
\usepackage{url}
\usepackage{tikz}
\usepackage{tabularx}
\usepackage{microtype}
\usepackage{hyperref,color}

\definecolor{webgreen}{rgb}{0,.35,0}
\definecolor{webbrown}{rgb}{.6,0,0}
\definecolor{RoyalBlue}{rgb}{0,0,0.9}
\definecolor{purp}{rgb}{0.6,0.05,0.8}
\definecolor{ora}{rgb}{0.7,0.35,0.02}

\newtheorem{theorem}{Theorem}[section]

\hypersetup{
   colorlinks=true, linktocpage=true, 
   urlcolor=webbrown, linkcolor=RoyalBlue, citecolor=webgreen,
   pdfauthor={Hugo Hiu Chak Cheng, Gary P. T. Choi},
   pdfsubject={Monotile kirigami}
}

\title{Monotile kirigami}
\author{Hugo Hiu Chak Cheng$^{1}$, Gary P. T. Choi$^{1,\ast}$\\
\\
\footnotesize{$^{1}$Department of Mathematics, The Chinese University of Hong Kong, Hong Kong, China}\\
\footnotesize{$^\ast$To whom correspondence should be addressed; E-mail: ptchoi@cuhk.edu.hk}
}
\date{ }

\begin{document}

\maketitle

\begin{abstract}
Kirigami, the art of paper cutting, has been widely used in the modern design of mechanical metamaterials. In recent years, many kirigami-based metamaterials have been designed based on different planar tiling patterns and applied to different science and engineering problems. However, it is natural to ask whether one can create deployable kirigami structures based on the simplest forms of tilings, namely the monotile patterns. In this work, we answer this question by proving the existence of periodic and aperiodic monotile kirigami structures via explicit constructions. In particular, we present a comprehensive collection of periodic monotile kirigami structures covering all 17 wallpaper groups and aperiodic monotile kirigami structures covering various quasicrystal patterns as well as polykite tilings. We further perform theoretical and computational analyses of monotile kirigami patterns in terms of their shape and size changes under deployment. Altogether, our work paves a new way for the design and analysis of a wider range of shape-morphing metamaterials.

\end{abstract}

%%%%%%%%%%%%%%%%%%%%%%%%%%%%%%%%%%%%%%%%%

\section{Introduction}

Kirigami, the traditional paper-cutting art, has inspired the design of many modern mechanical metamaterials~\cite{zhai2021mechanical,jiao2023mechanical,jin2024engineering,choi2024computational,dudek2025shape}. In particular, by introducing cuts on a sheet of materials, one can obtain deployable structures with different desired shape-morphing effects. Over the past few decades, kirigami metamaterials have been utilized in a wide range of practical applications, including soft electronics~\cite{blees2015graphene}, robotics~\cite{rafsanjani2018kirigami}, and other multifunctional devices~\cite{wang2023physics, yang2023new}.

Note that many kirigami structures are designed based on classical tiling patterns~\cite{grunbaum1986tilings}. In particular, there has been a vast number of studies on kirigami structures consisting of periodic tilings of simple polygons such as triangles~\cite{grima2006auxetic}, squares~\cite{grima2000auxetic}, rectangles~\cite{grima2004negative}, and rhombi~\cite{attard2008auxetic}. In recent years, some other designs of deployable kirigami structures have been proposed based on ancient artistic patterns~\cite{rafsanjani2016bistable} and more general wallpaper group patterns~\cite{grima2011auxetic,stavric2019geometrical,liu2021wallpaper,liu2024auxetic} as well as quasicrystal tilings~\cite{liu2022quasicrystal}. Besides utilizing standard polygonal tiling patterns for kirigami design, many prior works have also focused on generalizing the kirigami cut geometry and cut topology, in which either the tile shapes or the tile connections are carefully modified or optimized, to achieve different desired physical properties. For instance, Choi et al.~\cite{choi2019programming} developed a method for designing kirigami tessellations that can be transformed into a prescribed two-dimensional or three-dimensional shape. Additional geometric constraints were later identified for the design of compact reconfigurable kirigami structures~\cite{choi2021compact}. More recently, an additive approach was developed for the geometric design of kirigami structures~\cite{dudte2023additive}. As for the topological design, Lubbers and van Hecke~\cite{lubbers2019excess} developed a method to identify and count excess floppy modes in kirigami metamaterials. In~\cite{an2020programmable}, An et~al. designed kirigami patterns with a hierarchical topological structure and demonstrated their effectiveness in programming stress-strain responses and triggering a large variety of deformation patterns. Besides, Chen et al.~\cite{chen2020deterministic} developed deterministic and stochastic methods to control rigidity and connectivity of kirigami structures via changing their topologies. Later, methods for controlling the explosive rigidity percolation of kirigami structures have been further developed based on a change in kirigami topology~\cite{choi2023explosive}. Besides, the design of rotating unit mode metamaterials based on monotile systems has also been studied for a long time in perforated sheets~\cite{grima2010perforated}, hierarchical metamaterials~\cite{gatt2015hierarchical}, metallic periodic structures~\cite{taylor2014low}, auxetic metamaterials~\cite{lim2026auxetic}, and fragmentation-reconstitution metamaterials~\cite{lim2024metamaterial}. Other notable recent works on Euclidean tiling based metamaterials include the design of chiral honeycombs based on uniform Euclidean tessellations~\cite{mizzi2021chiralisation}, equiauxetic hinged Archimedean tilings~\cite{tarnai2022equiauxetic}, and transversely-isotropic chiral metamaterials with monohedral Euclidean pentagonal tessellations~\cite{mizzi2023design}.

In a recent groundbreaking work, Smith et al.~\cite{smith2024aperiodic} discovered aperiodic monotile patterns, which consist of a polykite shape called the ``Hat''. Moreover, it was shown that there is an infinite family of aperiodic monotiles generalized from the Hat pattern in the form of Tile$(a,b)$, where $a$ and $b$ are the relative lengths of two consecutive sides of the polykite tile. In particular, for nonzero $a,b$, Tile$(a,b)$ has sides with three possible lengths: $a,b,2a$. Tile$(0,1)$ is a union of four equilateral triangles and Tile$(1,0)$ is a union of eight equilateral triangles. More recently, they further discovered shapes that tile aperiodically without reflections~\cite{smith2024chiral}. In particular, the monotile pattern Tile (1,1) admits periodic tiling but can tile aperiodically without the use of reflected tiles. Also, the ``Spectre'' tile, which is modified from Tile (1,1), can tile aperiodically using translation and rotation only, whereas it is impossible to tile the plane using reflected and original tiles together. Since the discovery of aperiodic monotiles, many follow-up studies have further explored their engineering applications for the design of isotropic zero Poisson's ratio metamaterials~\cite{clarke2023isotropic}, high-performance composites~\cite{jung2024aperiodicity}, and studied their elastic properties~\cite{naji2024effective} and other physical properties such as chirality and zero modes~\cite{schirmann2024physical}.

Motivated by the above works, here we pose and solve the problem of designing deployable kirigami structures using \emph{periodic and aperiodic monotile patterns} (see Fig.~\ref{fig:F1}). In particular, we first establish the existence of periodic and aperiodic monotile kirigami patterns (monotile rotating unit-based metamaterials) via explicit constructions, with different design strategies utilized for creating representative monotile kirigami patterns. We then further perform theoretical and computational analyses on the shape and size changes of different monotile kirigami patterns under the deployment process. Altogether, our systematic study of monotile kirigami provides a solid theoretical foundation for designing and utilizing kirigami patterns for practical applications involving shape-morphing structures.

The rest of the paper is organized as follows. In Section~\ref{sect:main}, we first establish the existence of periodic monotile kirigami patterns, followed by the exploration of aperiodic monotile kirigami patterns. In Section \ref{sect:analysis}, we perform theoretical analysis on the shape change properties of both periodic and aperiodic monotile kirigami patterns. In Section~\ref{sect:numerical}, we further perform a computational analysis of the monotile kirigami patterns and assess their size change properties using numerical simulations. We conclude our work and discuss possible future directions in Section~\ref{sect:conclusion}.

\begin{figure}[t]
    \centering
    \includegraphics[width=0.8\linewidth]{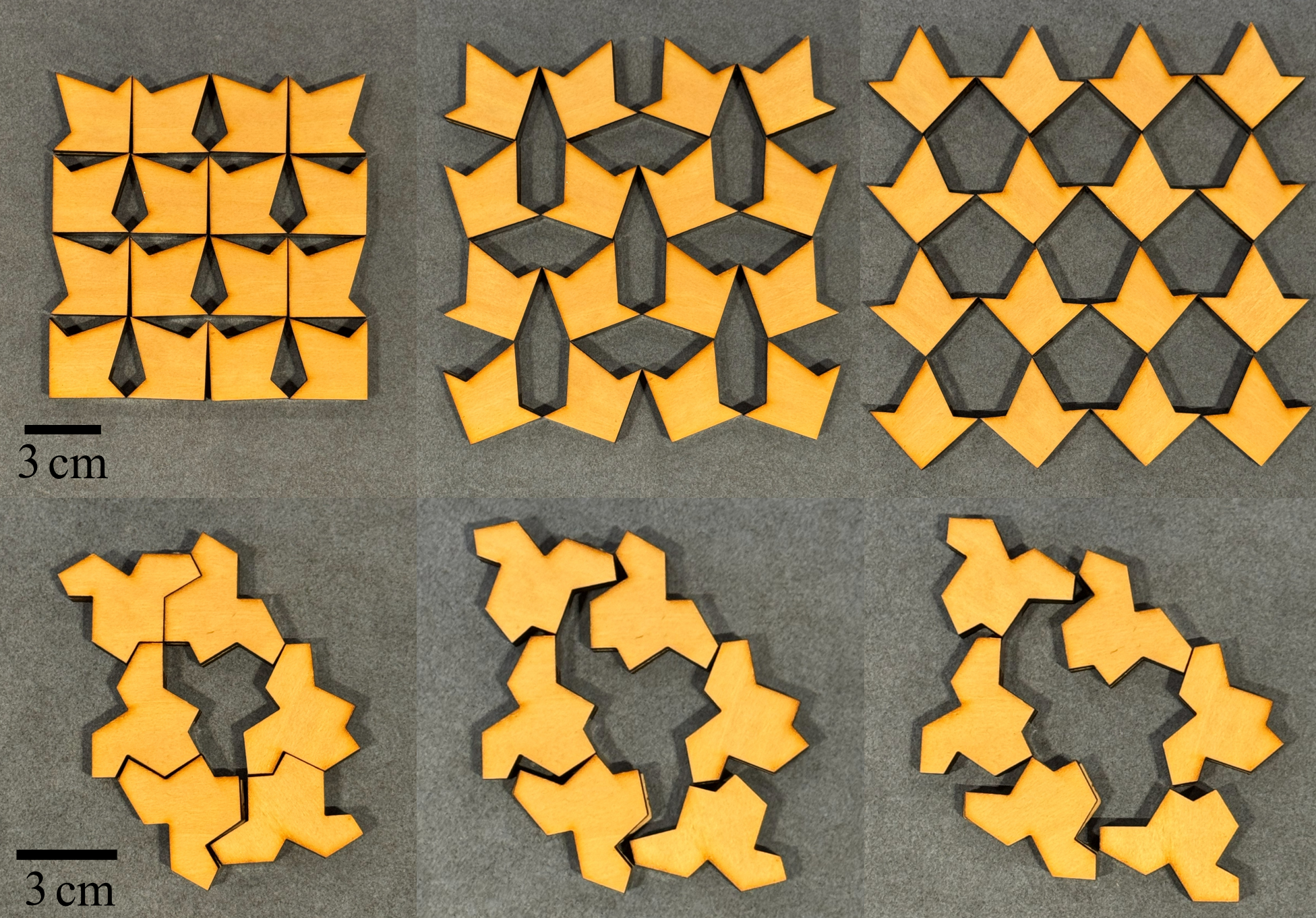}
    \caption{\textbf{Periodic and aperiodic monotile kirigami structures.} The top row shows the deployment of a physical model of a periodic monotile kirigami pattern (i.e., a periodic monotile rotating unit-based metamaterial) based on the cm wallpaper group constructed in our work (see also Fig.~\ref{fig:monotile_proof}). The bottom row shows the deployment of a physical model of an aperiodic monotile kirigami pattern (i.e., an aperiodic monotile rotating unit-based metamaterial) based on the ``Hat'' polykite tiles constructed in our work (see also Fig.~\ref{fig:polykite}). Both physical models consist of wood tiles produced by laser cutting and connected by tape joints.}
    \label{fig:F1}
\end{figure}

%%%%%%%%%%%%%%%%%%%%%%%%%%%%%%%%%%%%%%%%%%%%
\section{Existence of monotile kirigami patterns} \label{sect:main}

In this section, we establish the existence of deployable kirigami structures based on different monotile patterns. In particular, we will start with the case of periodic monotile patterns and create monotile kirigami patterns in all of the 17 wallpaper groups. We will then consider aperiodic monotile patterns based on quasicrystal tilings and polykite tilings and introduce our design strategies for achieving deployable kirigami structures.

\subsection{Periodic monotile kirigami}

As demonstrated in a prior work by Liu et~al.~\cite{liu2021wallpaper}, it is possible to design deployable kirigami patterns using any of the 17 wallpaper groups. It is natural to ask whether one can further design periodic \emph{monotile} kirigami patterns using any of the 17 wallpaper groups.

To be more precise, in this work, we define a \emph{periodic monotile kirigami pattern} as a pattern that satisfies all the following conditions:
\begin{itemize}
    \item The pattern is a rotating unit-based metamaterial that consists of only one type of tile. All tiles in the pattern are identical up to translation and rotation, and they are connected using point connections (i.e., with the connected tiles sharing a common vertex). Holes in the kirigami pattern are allowed.

    \item The kirigami pattern can be repeated infinitely via suitable translation and rotation.

    \item The kirigami pattern falls into one of the 17 wallpaper groups. Note that here the wallpaper group identification is solely based on the geometry of the tiles but not the connections between them.

    \item The pattern is deployable in the sense that it admits a continuous configuration path, so that it can undergo size/shape changes while preserving the designed connectivity.

    % \item 
\end{itemize}

Below, we show that it is indeed possible to construct monotile kirigami patterns in any of the 17 wallpaper groups by establishing the following constructive proof.

\begin{theorem}
    For each of the 17 wallpaper group patterns, there exists a deployable monotile kirigami pattern with a configuration belonging to that wallpaper group. 
\end{theorem}

\begin{figure}[t!]
    \centering
    \includegraphics[width=\linewidth]{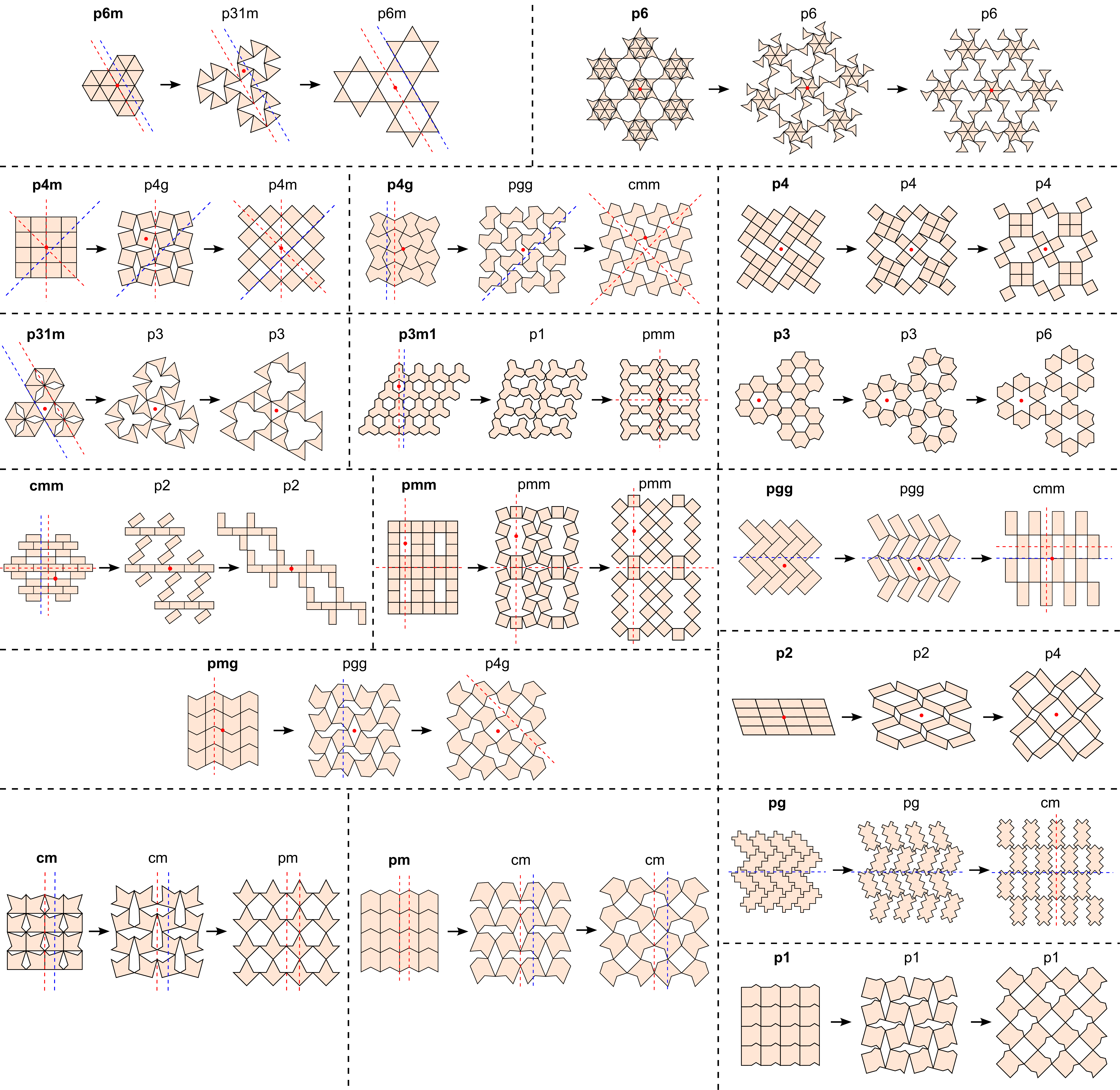}
    \caption{\textbf{Existence of deployable periodic monotile kirigami patterns for all 17 wallpaper groups.} Here, every red dot represents a rotation center, every red dashed line represents a reflectional axis, and every blue dashed line represents a glide reflectional axis. In all examples, the tile geometry and connectivity remain unchanged throughout the deployment.}
    \label{fig:monotile_proof}
\end{figure}

\begin{proof}
We prove the theorem by explicit construction. As shown in Fig.~\ref{fig:monotile_proof}, for each of the 17 wallpaper group patterns, we can indeed design a deployable monotile kirigami structure. Below, we describe the monotile patterns we identified from existing literature or newly obtained in this work based on some carefully designed construction strategies.

First, for the cases with 6-fold rotational symmetry, i.e., p6m and p6, we first note that the standard kagome pattern is already a monotile pattern in the p6m group (6-fold rotational symmetry together with reflectional symmetry), with the group change throughout the deployment process being p6m $\to$ p31m $\to$ p6m. As for the p6 case (without reflectional symmetry), here we design a 4-sided tile that looks like a triangle with a dent. By suitably rotating and connecting the tiles, we can break the reflectional symmetry while maintaining the 6-fold rotational symmetry and the deployability in forming a monotile kirigami structure. This structure remains to be p6 throughout the deployment process (p6 $\to$ p6 $\to$ p6).

We then consider the cases with 4-fold rotational symmetry (p4m, p4g, p4). For the p4m case (with mirrors at $45^\circ$), it is easy to see that the standard rotating squares pattern already satisfies the requirement (with the group change being p4m $\to$ p4g $\to$ p4m). As for the p4g case, here we design an 8-sided polygonal tile with a bow tie shape (similar to one of the patterns designed in~\cite{liu2021wallpaper}) and connect the tiles based on the standard rotating squares topology. We can then achieve a contracted state with reflectional symmetry but without mirrors at $45^\circ$, thereby getting a deployable monotile kirigami structure with group change p4g $\to$ pgg $\to$ cmm. Finally, for the p4 case (without reflectional symmetry), we achieve a deployable monotile kirigami pattern by carefully designing the connections between square tiles. In particular, each unit cell in the pattern consists of eight square tiles, where four of them are fully connected at the center and four of them are loosely connected at the four corners in a rotationally symmetric manner. This leads to a deployable monotile pattern staying in the p4 group throughout the deployment process (p4 $\to$ p4 $\to$ p4).

Next, we consider the three cases with 3-fold rotational symmetry (p31m, p3m1, p3). For the p31m case (with a rotation center off mirrors), we design a new monotile kirigami pattern using the 4-sided tile described above in the construction of the p6 pattern. This time, by suitably rotating the tiles and connecting them based on the standard kagome topology, we obtain a deployable structure with a p31m $\to$ p3 $\to$ p3 change throughout deployment, noting that the reflectional symmetry will be lost under deployment. For the p3m1 case, we design a new monotile pattern consisting of a Y-shaped tile. Note that connecting three such Y-shaped tiles in a way that forms a closed loop would make the pattern not deployable. Therefore, here we carefully design the connectivity between the tiles so that they loosely follow the rotating squares topology while having fewer connections. Under this design, we obtain a deployable structure that loses all its symmetry during the deployment but regains both rotational and reflectional symmetry (p3m1 $\to$ p1 $\to$ pmm). Lastly, for the p3 case, a monotile kirigami pattern can be constructed by using a tile that resembles a hexagon with a dent as designed in~\cite{liu2021wallpaper}, and the pattern remains in the p3 group under deployment (p3 $\to$ p3 $\to$ p3).

For the five cases with 2-fold symmetry (cmm, pmm, pmg, pgg, p2), we start by considering the cmm case. By using rectangular tiles with a special design of the connectivity, where some of the tiles are connected to the midpoint of some other tiles, we can design a deployable structure exhibiting perpendicular reflections and a rotation center off mirrors. Throughout the deployment, the reflectional symmetry is lost while the rotational symmetry is retained, yielding a cmm $\to$ p2 $\to$ p2 group change. For the pmm case, we design a structure by utilizing square tiles and connecting them in a way that forms rectangular holes. This gives a pattern with perpendicular reflections but without any rotation center off the mirrors. The pattern is deployable and remains in the pmm group throughout the deployment (pmm $\to$ pmm $\to$ pmm). For the pmg case, we design a 6-sided arrow-like tile and suitably arrange and connect the tiles with the standard rotating-squares topology. This gives a deployable structure with reflectional symmetry but not perpendicular reflections (and hence belongs to the pmg group). When deployed, the reflectional symmetry is lost but there is a glide reflectional symmetry. Moreover, when fully deployed, the pattern becomes a structure with 4-fold rotational symmetry (pmg $\to$ pgg $\to$ p4g). Next, for the pgg case and the p2 case, we note that the deployable kirigami patterns previously identified in~\cite{liu2021wallpaper} (pgg $\to$ pgg $\to$ cmm and p2 $\to$ p2 $\to$ p4) are already monotile patterns, and hence these two cases are easily fulfilled.

Finally, for the four cases without rotational symmetry (cm, pm, pg, p1), we can easily follow the patterns previously identified in~\cite{liu2021wallpaper} and obtain deployable monotile kirigami patterns for all four cases (cm $\to$ cm $\to$ pm, pm $\to$ cm $\to$ cm, pg $\to$ pg $\to$ cm, p1 $\to$ p1 $\to$ p1). Here, note that all four patterns are with the standard rotating-squares topology but different tile shapes that control the presence of reflections and glide reflections.

\end{proof}

\subsection{Aperiodic monotile kirigami}

Next, we study the existence of aperiodic monotile kirigami structures (i.e., aperiodic monotile rotating unit-based metamaterials) by exploring different aperiodic tiling patterns. 

Recall that in~\cite{liu2022quasicrystal}, Liu et~al. designed several deployable kirigami structures based on quasicrystal tilings. In particular, three major patterns were considered, namely the Penrose tiling, the Ammann--Beenker tiling, and the Stampfli tiling. The Penrose tiling consists of thinner rhombi and thicker rhombi and has 5-fold rotational symmetry. The Ammann--Beenker tiling consists of squares and rhombi and has 8-fold rotational symmetry. The Stampfli tiling consists of squares, triangles and rhombi and has 12-fold rotational symmetry. Also, to achieve deployability, the authors considered different construction approaches, namely the addition method, removal method, and Hamiltonian method. Specifically, the addition method involves inserting thin tiles in between the existing tiling patterns to change the connectivity and increase the flexibility of the structures. The removal method involves removing certain tiles from the existing tiling patterns, creating holes and gaps to enable deployability. The Hamiltonian method keeps all tiles and focuses on designing a Hamiltonian cycle connecting all tiles, thereby yielding a deployable structure.

\begin{figure}[t!]
    \centering
    \includegraphics[width=\linewidth]{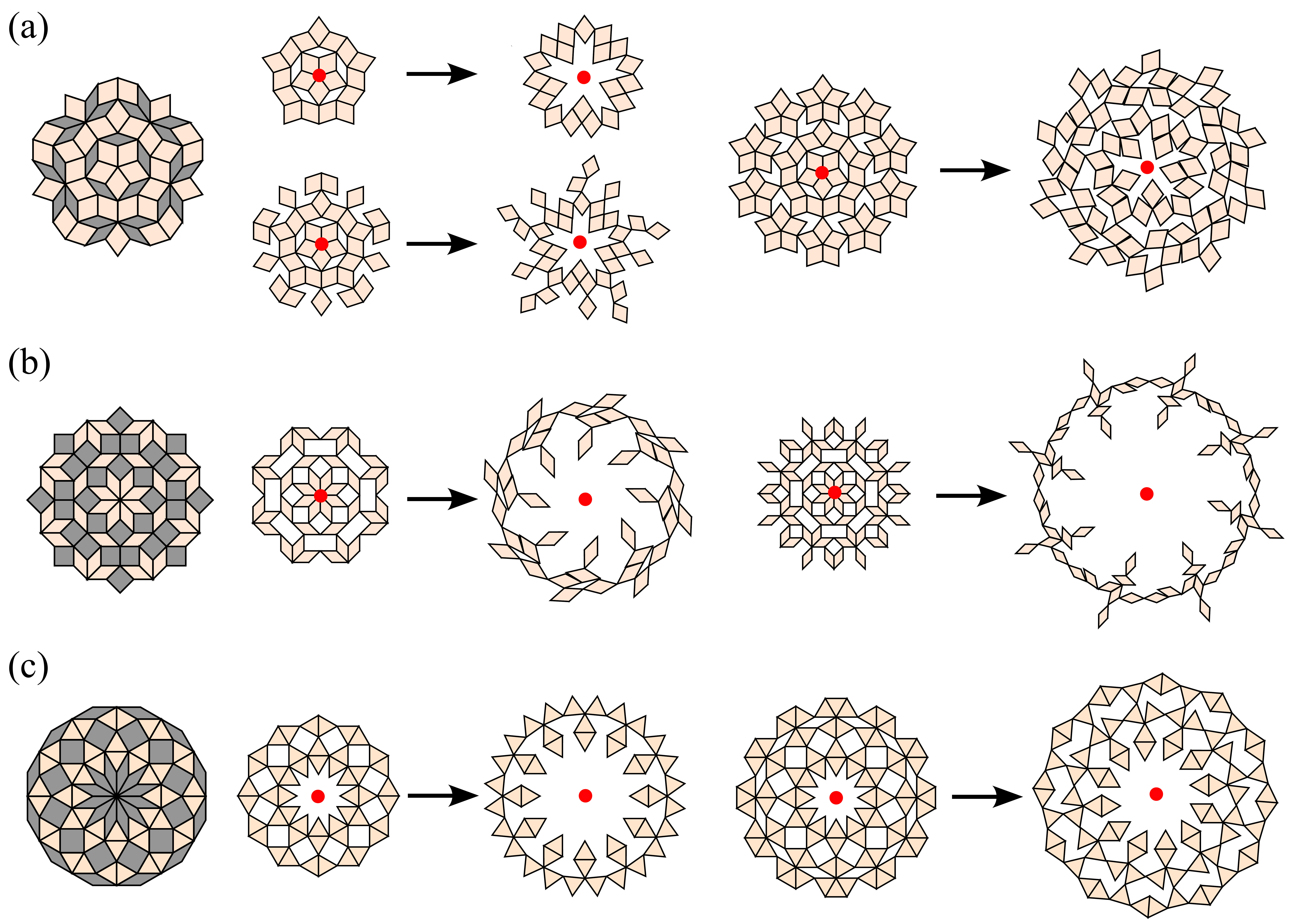}
    \caption{\textbf{Aperiodic monotile kirigami patterns constructed using quasicrystal tilings.} (a) For the 5-fold Penrose tiling, by removing the thinner rhombi (highlighted in grey) as described in~\cite{liu2022quasicrystal}, we can achieve deployable monotile kirigami patterns with different resolutions. (b) For the 8-fold Ammann--Beenker tiling, by removing the square tiles (highlighted in grey) as described in~\cite{liu2022quasicrystal}, we can achieve deployable monotile kirigami patterns with different resolutions. (c) For the 12-fold Stampfli tiling, here we propose removing both the rhombi and squares (highlighted in grey) and keeping only the triangles. This gives deployable monotile kirigami structures with different resolutions. In all examples, the red circle dots show the center of the rotational symmetry. The tile geometry and connectivity remain unchanged throughout the deployment.}
    \label{fig:quasicrystal}
\end{figure}

Note that while the above patterns are aperiodic kirigami patterns, many of them involve multiple types of tiles. Consequently, the addition method and Hamiltonian method described above are not suitable for our construction of monotile kirigami structures. Therefore, here we utilize the removal method and focus on removing tiles from these patterns to achieve the monotile property while preserving deployability. In particular, for the 5-fold Penrose tiling, it was pointed out in~\cite{liu2022quasicrystal} that deployable kirigami structures can be designed by removing all thinner rhombi. Note that the resulting structures are monotile patterns and hence satisfy our requirement (see Fig.~\ref{fig:quasicrystal}(a)). Similarly, for the 8-fold Ammann--Beenker tiling, we can remove all square tiles as described in~\cite{liu2022quasicrystal} to achieve deployable monotile kirigami structures (see Fig.~\ref{fig:quasicrystal}(b)). As for the 12-fold Stampfli tiling, since it consists of three types of tiles (triangles, squares, and rhombi), getting a monotile pattern would require removing two types of tiles. Here, we propose to remove both the squares and rhombi, keeping only the triangles. Then, by suitably connecting the triangles, we can achieve deployable monotile kirigami patterns (see Fig.~\ref{fig:quasicrystal}(c)). For all patterns, the tile connectivity remains unchanged throughout the deployment.

\begin{figure}[t!]
    \centering
    \includegraphics[width=\linewidth]{Figure4.pdf}
    \caption{\textbf{Aperiodic monotile kirigami patterns constructed using polykite monotiles.} (a) The left panel shows an illustration of the geometry of a polykite monotile Tile$(a,b)$, which has side lengths $a,b,2a$. The right panel shows an aperiodic tiling with the ``Hat'' monotile generated using the first level of the $H_7$ substitution rule~\cite{smith2024aperiodic}. The red lines show the lattice representation of the kirigami structure if we consider adding a connection between all tiles that share an edge. The existence of triangles in the lattice representation indicates that we need to remove the central tile (in grey). (b) By removing the central tile and connecting the remaining 6 tiles, we can achieve a deployable polykite kirigami structure. Here, the red dots indicate the point connections between the tiles. (c) A larger ``Hat'' tiling generated using the second level of the $H_7$ rule, with a total of 47 tiles. Here, the grey tiles are again to be removed to achieve deployability. The tiles that do not share any common vertex with the grey tiles (highlighted in red) are also to be removed. (d) By removing all the grey and red tiles and connecting the remaining 36 tiles, we can again achieve a deployable kirigami structure. (e) One can also construct monotile kirigami patterns using other polykite geometries, such as ``Turtle'' and other variants~\cite{smith2024aperiodic} in the form of Tile$(a,b)$, where $a$ and $b$ are the lengths of two consecutive sides of the tile. Here, we present several examples (top) and their deployed configurations (bottom). Here, the tiles are connected by point connections as shown in (b). In all examples, the tile geometry and connectivity remain unchanged throughout the deployment.}
    \label{fig:polykite}
\end{figure}

Besides, with the recent discovery of the aperiodic polykite monotile patterns~\cite{smith2024aperiodic}, it is natural to ask whether they can be used to construct deployable kirigami patterns. Here we consider generating various aperiodic polykite monotile patterns, including ``Hat'', ``Turtle'', and other variants Tile$(a,b)$ with different choices of $a$ and $b$, using the generator in \url{https://cs.uwaterloo.ca/~csk/hat/h7h8.html} with the ``$H_7$ substitution rule''.

First, note that the first level of the $H_7$ rule gives 7 polykite tiles. It is noteworthy that if we consider adding a connection between all tiles that share an edge, then the lattice representation of such connections will involve a triangle (see the red edges in Fig.~\ref{fig:polykite}(a)). Therefore, if we connect all 7 tiles, the overall structure will be rigid. Also, all supertiles of the polykite tiles generated based on the $H_7$ rule include these 7 tiles. Hence, to make a deployable structure based on the polykite tiles, we cannot keep all the tiles in our construction. This motivates us to again consider using the removal method as described in~\cite{liu2022quasicrystal} for constructing polykite kirigami patterns consisting of Hat or its variants. More specifically, here we remove the central tile reflected hat so that there are 6 tiles left. Then, we add a connection for every pair of neighboring tiles. This yields a deployable kirigami structure as shown in Fig.~\ref{fig:polykite}(b). 

Next, the second level of the $H_7$ rule gives a total of 47 tiles as shown in Fig.~\ref{fig:polykite}(c). Analogous to the first level, here we remove the central tiles in each cell (highlighted in grey) which are all reflected hats. Moreover, we remove the five tiles that do not share any common vertex with the grey tiles (highlighted in red) which are type T according to the classification of Smith et al., leaving a total of 36 tiles. We can then build the connections between neighboring cells, thereby getting a deployable structure as shown in Fig.~\ref{fig:polykite}(d). By repeating this process, we can extend our construction for even larger patterns. We remark that the above-mentioned construction method is applicable not only for the ``Hat'' tile but also for ``Turtle'' and many other variants of polykite tiles as shown in Fig.~\ref{fig:polykite}(e). These polykite patterns remain connected after deployment.

%%%%%%%%%%%%%%%%%%%%%%%%%%%%%%%%%%%%%%%%%
\section{Theoretical analysis of monotile kirigami patterns} \label{sect:analysis}

After proving the existence of monotile kirigami patterns for both the periodic and aperiodic cases, we perform a systematic analysis of their geometric properties. In this section, we focus on the shape-related properties of them under the deployment. Specifically, we analyze the gain, loss, and preservation of the symmetries of the monotile kirigami patterns and establish several theoretical results.

\subsection{Symmetry change of periodic monotile kirigami patterns}

Note that periodic tiling patterns are characterized by their rotational, reflectional, and glide reflectional symmetries. For periodic monotile kirigami patterns, it is natural to ask whether some or all of these symmetries may be changed or preserved throughout the deployment process. In fact, the answer is affirmative as stated in the following theorem: 

\begin{theorem}
Gain, loss, and preservation of rotational, reflectional, and glide reflectional symmetries are possible for monotile kirigami patterns.
\end{theorem}

\begin{proof}
We prove the statement by identifying explicit examples for different cases. 
    
For rotational symmetry, note that the p2 monotile kirigami pattern in Fig.~\ref{fig:monotile_proof} illustrates a gain in rotational symmetry from 2 to 4 throughout the deployment process. As for loss of rotational symmetry, the p4g monotile kirigami pattern in Fig.~\ref{fig:monotile_proof} shows a loss in rotational symmetry from 4 to 2 throughout the deployment process. Besides, the p4 monotile kirigami pattern in Fig.~\ref{fig:monotile_proof} shows a preservation of rotational symmetry throughout the deployment process.

For reflectional symmetry, note that the pg monotile kirigami pattern in Fig.~\ref{fig:monotile_proof} illustrates a gain in reflectional symmetry throughout the deployment process. As for loss of reflectional symmetry, the p31m monotile kirigami pattern in Fig.~\ref{fig:monotile_proof} shows a loss in reflectional symmetry throughout the deployment process. Besides, the p4m monotile kirigami pattern in Fig.~\ref{fig:monotile_proof} shows a preservation of reflectional symmetry throughout the deployment process.

For glide reflectional symmetry, here we focus exclusively on glide reflection axes off mirrors so that the analysis is nontrivial (otherwise, any reflection axis can already serve as a glide reflection axis). Note that the pm monotile kirigami pattern in Fig.~\ref{fig:monotile_proof} illustrates a gain in glide reflectional symmetry throughout the deployment process. As for loss of glide reflectional symmetry, the cm monotile kirigami pattern in Fig.~\ref{fig:monotile_proof} shows a loss in glide reflectional symmetry throughout the deployment process. Besides, the p4m monotile kirigami pattern in Fig.~\ref{fig:monotile_proof} shows a preservation of glide reflectional symmetry throughout the deployment process.

Altogether, the above analysis shows that gain, loss, and preservation of all three kinds of symmetry are possible for monotile kirigami patterns.

\end{proof}

While the above results show that the gain, loss, and preservation of symmetry are all possible, one may be interested in further exploring the space of all possible symmetry changes more precisely. In the rest of this subsection, we further analyze specific cases of gain, loss, and preservation of symmetries.

\begin{figure}[t!]
    \centering
    \includegraphics[width=\linewidth]{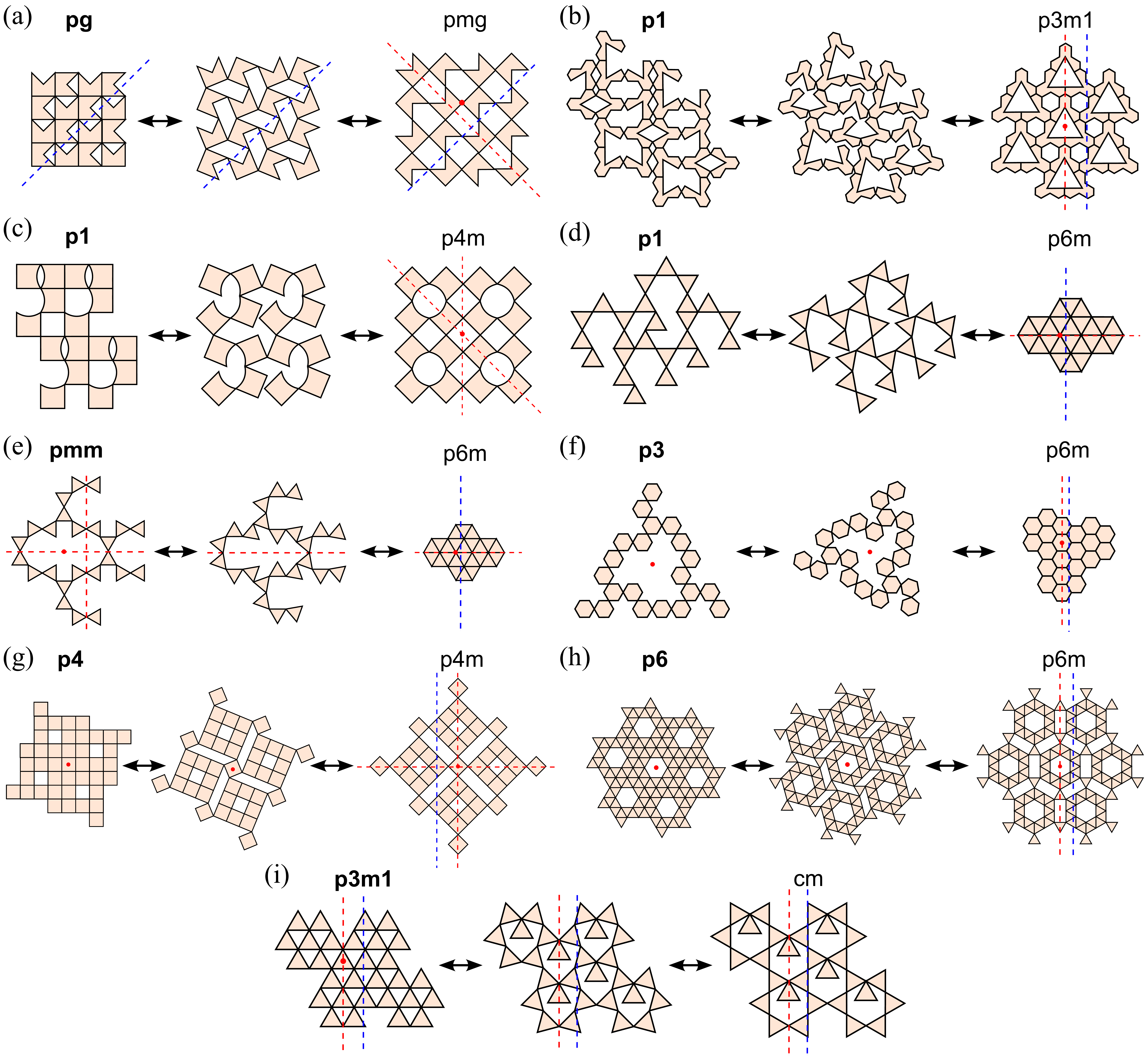}
    \caption{\textbf{Additional examples for gain, loss, and preservation of symmetry in periodic monotile kirigami structures.} (a) A pg $\leftrightarrow$ pmg pattern. (b) A p1 $\leftrightarrow$ p3m1 pattern. (c) A p1 $\leftrightarrow$ p4m pattern. (d) A p1 $\leftrightarrow$ p6m pattern. (e) A pmm $\leftrightarrow$ p6m pattern. (f)~A p3 $\leftrightarrow$ p6m pattern. (g) A p4 $\leftrightarrow$ p4m pattern. (h) A p6 $\leftrightarrow$ p6m pattern. (i) A p3m1 $\leftrightarrow$ cm pattern. Here, every red dot represents a rotation center, every red dashed line represents a reflectional axis, and every blue dashed line represents a glide reflectional axis. In all examples, the tile geometry and connectivity remain unchanged throughout the deployment.}
    \label{fig:monotile_properties_gain}
\end{figure}

\subsubsection{Gain of symmetry}
We start by analyzing the possible gain of symmetry under the deployment of monotile kirigami patterns.

First, for rotational symmetry, we identify the following examples to show the possibility of having any $m$-fold to $n$-fold rotational symmetry gain, where $m,n = 1,2,3,4,6$ with $m < n$ and $n$ is divisible by $m$ ($m|n$):
\begin{itemize}
    \item $1\to 2$: The pg $\to$ pmg pattern in Fig.~\ref{fig:monotile_properties_gain}(a). 
    \item $1\to 3$: The p1 $\to$ p3m1 pattern in Fig.~\ref{fig:monotile_properties_gain}(b). 
    \item $1\to 4$: The p1 $\to$ p4m pattern in Fig.~\ref{fig:monotile_properties_gain}(c) (previously identified in~\cite{liu2021wallpaper}).   
    \item $1\to 6$: The p1 $\to$ p6m pattern in pattern in Fig.~\ref{fig:monotile_properties_gain}(d) (previously identified in~\cite{liu2021wallpaper}). 
    \item $2\to 4$: The p2 $\to$ p4 pattern in Fig.~\ref{fig:monotile_proof}.
    \item $2\to 6$: The pmm $\to$ p6m pattern in Fig.~\ref{fig:monotile_properties_gain}(e) (previously identified in~\cite{liu2021wallpaper}). 
    \item $3\to 6$: The p3 $\to$ p6 pattern in Fig.~\ref{fig:monotile_proof}.
\end{itemize}

\noindent For reflectional symmetry, for each $n = 1,2,3,4,6$, we identify an example of a deployable monotile kirigami pattern with $n$-fold rotational symmetry that gains its reflectional symmetry after the deployment:

\begin{itemize}
\item $n=1$: The pg $\to$ cm pattern in Fig.~\ref{fig:monotile_proof}. Note that the red dotted line at the final state shows the gained reflectional symmetry axis.

\item $n=2$: The pgg $\to$ cmm pattern in Fig.~\ref{fig:monotile_proof}. Note that the red dotted lines at the final state show the gained reflectional symmetry axes.

\item $n=3$: The p3 $\to$ p6m pattern in Fig.~\ref{fig:monotile_properties_gain}(f) (previously identified in~\cite{liu2021wallpaper}). Note that the red dotted line at the final state shows the gained reflectional symmetry axis.

\item $n=4$: The p4 $\to$ p4m pattern in Fig.~\ref{fig:monotile_properties_gain}(g). Note that the red dotted lines at the final state show the gained reflectional symmetry axes.

\item $n=6$: The p6 $\to$ p6m pattern in Fig.~\ref{fig:monotile_properties_gain}(h).  Note that the red dotted line at the final state shows the gained reflectional symmetry axis.

\end{itemize}

\noindent For glide reflectional symmetry, for each $n = 1,2,3,4,6$, we identify an example of a deployable monotile kirigami pattern with $n$-fold rotational symmetry that gains its glide reflectional symmetry after the deployment (again, here we focus on non-trivial glide reflection axes off mirrors):

\begin{itemize}
\item $n=1$: The p1 $\to$ p6m pattern in Fig.~\ref{fig:monotile_properties_gain}(d). Note that the blue dotted line at the final state shows the gained glide reflectional symmetry axis. 

\item $n=2$: The pmm $\to$ p6m pattern in Fig.~\ref{fig:monotile_properties_gain}(e). Note that the blue dotted line at the final state shows the gained glide reflectional symmetry axis.

\item $n=3$: The p3 $\to$ p6m pattern in Fig.~\ref{fig:monotile_properties_gain}(f). Note that the blue dotted line at the final state shows the gained glide reflectional symmetry axis.

\item $n=4$: The p4 $\to$ p4m pattern in Fig.~\ref{fig:monotile_properties_gain}(g). Note that the blue dotted line at the final state shows the gained glide reflectional symmetry axis.

\item $n=6$: The p6 $\to$ p6m pattern in Fig.~\ref{fig:monotile_properties_gain}(h). Note that the blue dotted line at the final state shows the gained glide reflectional symmetry axis.

\end{itemize}

\subsubsection{Loss of symmetry}
\noindent For rotational symmetry, we identify the following examples to show the possibility of having any $m$-fold $\to$ $n$-fold rotational symmetry loss, where $m,n = 1,2,3,4,6$ with $m > n$ and $n|m$:

\begin{itemize}

    \item $2\to 1$: The pmg $\to$ pg pattern in Fig.~\ref{fig:monotile_properties_gain}(a).
    
    \item $3\to 1$: The p3m1 $\to$ p1 pattern in Fig.~\ref{fig:monotile_proof}.
    
    \item $4\to 1$: The p4m $\to$ p1 pattern in Fig.~\ref{fig:monotile_properties_gain}(c).
    \item $4\to 2$: The p4g $\to$ cmm pattern in Fig.~\ref{fig:monotile_proof}.
    
    \item $6\to 1$: The p6m $\to$ p1 pattern in Fig.~\ref{fig:monotile_properties_gain}(d).
    \item $6\to 2$: The p6m $\to$ pmm pattern in Fig.~\ref{fig:monotile_properties_gain}(e).
    \item $6\to 3$: The p6m $\to$ p3 pattern in Fig.~\ref{fig:monotile_properties_gain}(f).

\end{itemize}

\noindent For reflectional symmetry, for each $n = 1,2,3,4,6$, we identify an example of a deployable monotile kirigami pattern with $n$-fold rotational symmetry that loses its reflectional symmetry after the deployment:

\begin{itemize}
    \item $n=1$: The pmg $\to$ pg pattern in Fig.~\ref{fig:monotile_properties_gain}(a) (the red dotted line).
    \item $n=2$: The cmm $\to$ p2 pattern in Fig.~\ref{fig:monotile_proof} (the red dotted lines).
    \item $n=3$: The p31m $\to$ p3 pattern in Fig.~\ref{fig:monotile_proof} (the red dotted line).
    \item $n=4$: The p4m $\to$ p1 pattern in Fig.~\ref{fig:monotile_properties_gain}(c) (the red dotted lines). 
    \item $n=6$: The p6m $\to$ p1 pattern in Fig.~\ref{fig:monotile_properties_gain}(d) (the red dotted line).
\end{itemize}

\noindent For glide reflectional symmetry, for each $n = 1,2,3,4,6$, we identify an example of a deployable monotile kirigami pattern with $n$-fold rotational symmetry that loses its glide reflectional symmetry after the deployment:
\begin{itemize}
\item $n=1$: The cm $\to$ pm pattern in Fig.~\ref{fig:monotile_proof} (the blue dotted line). 

\item $n=2$: The cmm $\to$ p2 pattern in Fig.~\ref{fig:monotile_proof} (the blue dotted line). 

\item $n=3$: The p31m $\to$ p3 pattern in Fig.~\ref{fig:monotile_proof} (the blue dotted line). 

\item $n=4$: The p4m $\to$ p4 pattern in Fig.~\ref{fig:monotile_properties_gain}(g) (the blue dotted line).

\item $n=6$: The p6m $\to$ p1 pattern in Fig.~\ref{fig:monotile_properties_gain}(d) (the blue dotted line).

\end{itemize}

\subsubsection{Preservation of symmetry}
\noindent For rotational symmetry, we identify the following examples to show the possibility of preserving the $n$-fold rotational symmetry throughout the entire deployment process, where $n = 1,2,3,4,6$:

\begin{itemize}

    \item $1\to 1$: The p1 $\to$ p1 pattern in Fig.~\ref{fig:monotile_proof}.
    \item $2\to 2$: The cmm $\to$ p2 pattern in Fig.~\ref{fig:monotile_proof}.
    \item $3\to 3$: The p31m $\to$ p3 pattern in Fig.~\ref{fig:monotile_proof}. 
    \item $4\to 4$: The p4 $\to$ p4 pattern in Fig.~\ref{fig:monotile_proof}.    
    \item $6\to 6$: The p6 $\to$ p6 pattern in Fig.~\ref{fig:monotile_proof}.

\end{itemize}

\noindent For reflectional symmetry, for each $n = 1,2,3,4,6$, we identify an example of a deployable monotile kirigami pattern with $n$-fold rotational symmetry that preserves its reflectional symmetry after the deployment:
\begin{itemize}
    \item $n = 1$: The cm $\to$ pm pattern in Fig.~\ref{fig:monotile_proof} (the red dotted line).
    \item $n = 2$: The pmm $\to$ pmm pattern in Fig.~\ref{fig:monotile_proof} (the red dotted line).
    \item $n = 3$: The p3m1 $\to$ cm pattern in Fig.~\ref{fig:monotile_properties_gain}(i) (the red dotted line).
    \item $n = 4$: The p4m $\to$ p4m pattern in Fig.~\ref{fig:monotile_proof} (the red dotted line). 
    \item $n = 6$: The p6m $\to$ p6m pattern in Fig.~\ref{fig:monotile_proof} (the red dotted line).
    
\end{itemize}

\noindent For glide reflectional symmetry, for each $n = 1,2,3,4,6$, we identify an example of a deployable monotile kirigami pattern with $n$-fold rotational symmetry that preserves its glide reflectional symmetry throughout the deployment:
\begin{itemize}
    \item $n = 1$: The pg $\to$ cm pattern in Fig.~\ref{fig:monotile_proof} (the blue dotted line).
    \item $n = 2$: The pgg $\to$ cmm pattern in Fig.~\ref{fig:monotile_proof} (the blue dotted line).
    \item $n = 3$: The p3m1 $\to$ cm pattern in Fig.~\ref{fig:monotile_properties_gain}(i) (the blue dotted line).
    \item $n = 4$: The p4m $\to$ p4m pattern in Fig.~\ref{fig:monotile_proof} (the blue dotted line). 
    \item $n = 6$: The p6m $\to$ p6m pattern in Fig.~\ref{fig:monotile_proof} (the blue dotted line).

\end{itemize}

\subsubsection{Summary}
From the above analysis, we can summarize the possible symmetry changes as follows:

\begin{theorem}
    For any $(m,n)\in {1,2,3,4,6}$ and $m|n$ or $n|m$, it is possible to design a deployable monotile kirigami pattern that achieves an $m$-fold to $n$-fold rotational symmetry change.

\end{theorem}

\begin{theorem}
    For $n=1,2,3,4,6$, it is possible to design a deployable monotile kirigami pattern with $n$-fold rotational symmetry that achieves any target reflectional symmetry change (gain/loss/preservation).
\end{theorem}

\begin{theorem}
    For $n=1,2,3,4,6$, it is possible to design a deployable monotile kirigami pattern with $n$-fold rotational symmetry that achieves any target glide reflectional symmetry change (gain/loss/preservation).
\end{theorem}

We remark that there exist other special symmetry changes not covered by the above theorems. For instance, the p3m1$\to$ pmm pattern in Fig.~\ref{fig:monotile_proof} exhibits a $3$-fold to $2$-fold rotationally symmetry change.

\subsection{Symmetry change of aperiodic monotile kirigami patterns}

It is natural to ask whether we can further analyze the theoretical properties of aperiodic monotile kirigami patterns. While the periodic patterns can be analyzed based on the theory of wallpaper groups, it is much more difficult to analyze the aperiodic kirigami patterns systematically due to the relatively limited understanding of aperiodic tilings and their constructions in general.

Here, we first consider the rotational symmetry change of the aperiodic monotile kirigami patterns that we have identified. For the Penrose-based monotile kirigami pattern, we see that the 5-fold symmetry can be preserved throughout the deployment process. Similarly, for the Ammann--Beenker-based monotile kirigami pattern, the 8-fold symmetry can be preserved under deployment, and the 12-fold symmetry of the Stampfli-based monotile kirigami pattern can also be preserved under deployment. As for the polykite monotile kirigami patterns and the variants, no rotational symmetry can be preserved under deployment in general.

Besides, note that while the aperiodic patterns lack translational symmetry, we may still study their properties with respect to reflectional symmetry. First, we note that for the Penrose-based monotile kirigami patterns, the initial state contains 5 reflection axes. However, during deployment, all reflectional symmetries will be lost. Similarly, while the Ammann--Beenker-based monotile kirigami patterns have 8 reflection axes at the initial state, all the reflectional symmetries will be lost under the deployment. As for the stampfli-based monotile kirigami patterns, we can also see that there are 12 reflection axes at the initial state, but all reflectional symmetries will be lost as we consider sufficiently high resolutions. As for the polykite monotile kirigami patterns and variants, they do not exhibit any reflectional symmetry in either the contracted or deployed state.

Altogether, for aperiodic monotile kirigami patterns, our analysis suggests that while the rotational symmetry of them can be preserved throughout the deployment process for some cases, the reflectional symmetry present at the initial contracted state will be lost under the deployment in general. This shows a significant difference between periodic and aperiodic monotile kirigami patterns.

%%%%%%%%%%%%%%%%%%%%%%%%%%%%%%%%%%%%%%%%%

\section{Computational analysis of monotile kirigami patterns} \label{sect:numerical}

Note that the above analysis on the gain, loss, and preservation of symmetry in monotile kirigami patterns has only focused on the shape change of them throughout deployment. An equally important aspect of these patterns is their size change, which may largely affect the suitability of different kirigami structures for certain applications. In this section, we perform a computational analysis of the monotile kirigami patterns we designed or identified. Specifically, we utilize the 2D kirigami deployment simulator~\cite{liu2022quasicrystal} and perform deployment simulations of different patterns, from which we further extract the size change measures and analyze them. 

More specifically, for each kirigami structure, we set radial springs on certain boundary vertices to pull them outward. Following the setup in~\cite{liu2022quasicrystal}, the radial springs are attached to a circle centered at the pattern center, and the circle is large enough so that the deployed pattern can never reach it. The overall structure is then deployed continuously under rigid-body motion in the simulator as the radial springs pull their connected tiles outward, with all tile geometry and connectivity preserved. The deployment is then completed as the overall structure reaches a steady state, i.e., the radial springs cannot pull their connected tiles any further outward. See Fig.~\ref{fig:deployment} for an example of the deployment process of the cm$\to$cm$\to$pm pattern achieved by the simulator.

\begin{figure}[t]
    \centering
    \includegraphics[width=\linewidth]{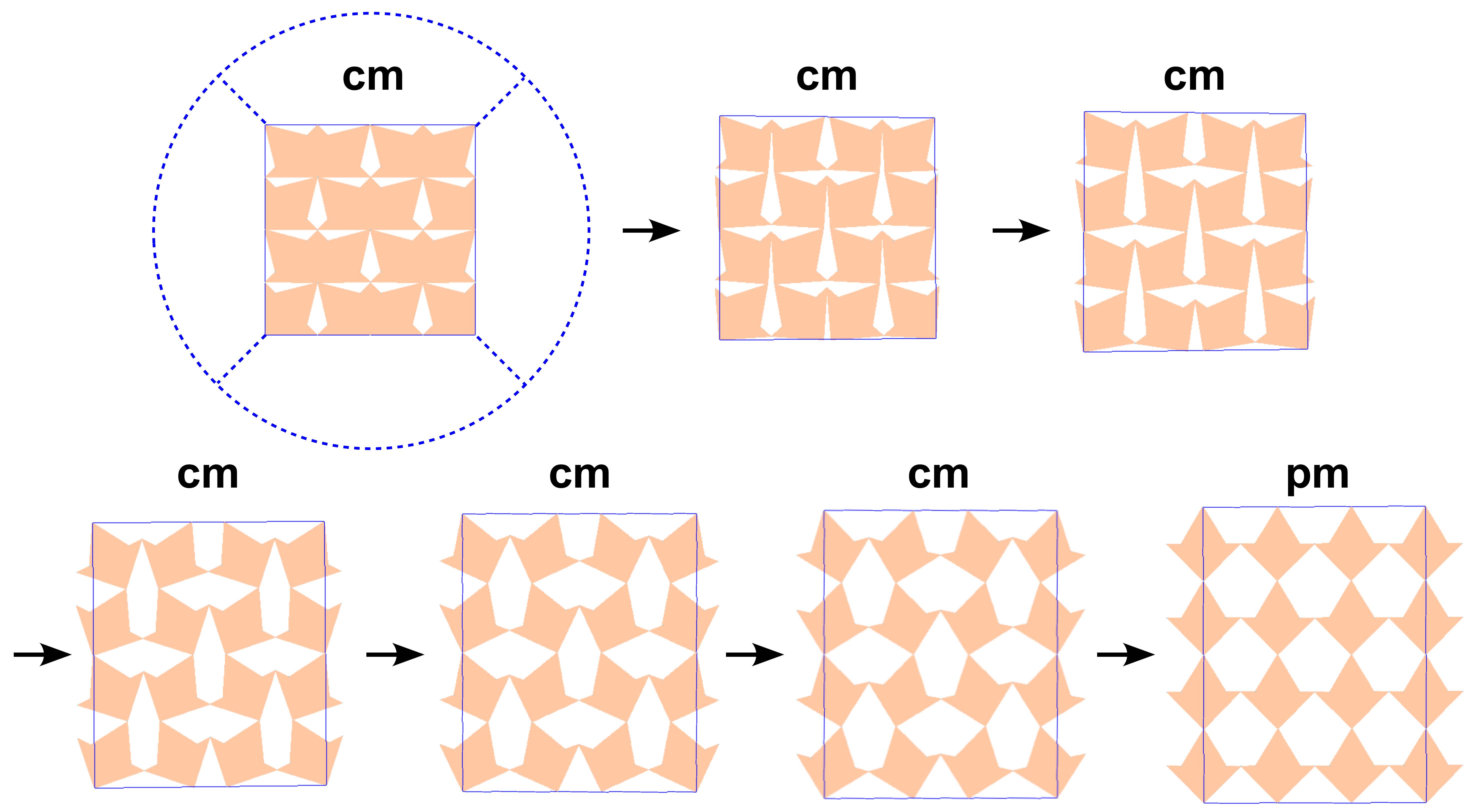}
    \caption{\textbf{The simulated deployment process of the cm$\to$cm$\to$pm monotile kirigami pattern obtained using the deployment simulator~\cite{liu2022quasicrystal}.} We start by picking the four corners of the initial cm pattern as the hull vertices, forming a convex hull in blue. Radial springs attached to a sufficiently large circle (indicated by the blue dashed lines and circle) are then used to drive the deployment of the structure until it reaches a steady state.}
    \label{fig:deployment}
\end{figure}

\subsection{Numerical results on the geometric change of monotile kirigami patterns under deployment}

Here, to assess the size change ratio of the kirigami pattern under deployment, we first consider the convex hull of certain boundary vertices in the input configuration of the pattern. We can then consider the ratio
\begin{equation}
    \text{SCR} = \frac{\text{Area($H_{\text{final}}$)}}{\text{Area($H_{\text{initial}}$)}},
\end{equation}
where $H_{\text{initial}}$ is the region formed by the convex hull vertices at the initial state, and $H_{\text{final}}$ is the region formed by the corresponding vertices after the deployment. Here, the area of a polygon with vertices $(x_1, y_1)$, $\dots$, $(x_n, y_n)$ can be given by the shoelace formula:
\begin{equation}
    \text{Area} = \frac{1}{2}\left|\sum_{i=1}^n (x_i y_{i+1} - x_{i+1}y_i)\right|,
\end{equation}
where $x_{n+1} = x_1$ and $y_{n+1} = y_1$.

\begin{table}[t]
    \centering
    \begin{tabular}{c|c|c}
        \textbf{Pattern} & \textbf{Size change ratio} & \textbf{Perimeter change ratio} \\ \hline 
         p6m $\to$ p31m $\to$ p6m & 2.8299 &1.7480\\
         p6 $\to$ p6 $\to$ p6 &1.2902 &1.1359\\
         p4m $\to$ p4g $\to$ p4m  & 1.5412 & 1.2501 \\
         p4g $\to$ pgg $\to$ cmm &1.4811 & 1.2609\\
         p4 $\to$ p4 $\to$ p4 &1.2804 &1.1882\\
         
        p31m $\to$ p3 $\to$ p3 & 1.8056 &1.5691\\
        p3m1 $\to$ p1 $\to$ pmm &1.1472 &1.0519\\
        p3 $\to$ p3 $\to$ p6  &1.2132 &1.1297\\
        
        cmm $\to$ p2 $\to$ p2 & 1.6449 & 1.4569\\
         pmm $\to$ pmm $\to$ pmm & 1.7319 & 1.3255\\
         
         pmg $\to$ pgg $\to$ p4g & 1.3744 & 1.2391\\
         pgg $\to$ pgg $\to$ cmm & 1.5437 & 1.3084\\
         p2 $\to$ p2 $\to$ p4  & 1.8382 & 1.1142\\
          cm $\to$ cm $\to$ pm & 1.5408 & 1.2492\\
         pm $\to$ cm $\to$ cm & 1.4436 & 1.2415\\
         pg $\to$ pg $\to$ cm  & 1.5304 & 1.2784\\
         
         p1 $\to$ p1 $\to$ p1 & 1.5014 & 1.2395\\

    \end{tabular}
    \caption{Analysis of the size change ratio and perimeter change ratio for different periodic monotile kirigami patterns.}
    \label{tab:patterns}
\end{table}

Similarly, using the convex hull vertices, we can also consider the perimeter change ratio (PCR) of the monotile kirigami patterns under deployment:
\begin{equation}
    \text{PCR}=\frac{\text{Perimeter($H_{\text{final}}$)}}{\text{Perimeter($H_{\text{initial}}$)}},
\end{equation}
where the perimeter of a polygon with vertices $(x_1, y_1)$, $\dots$, $(x_n, y_n)$ can be calculated by 
\begin{equation}
    \text{Perimeter} = \sum_{i=1}^n \sqrt{(x_i - x_{i+1})^2 +(y_i-y_{i+1})^2}.
\end{equation} 
Note that both of the above SCR and PCR values are calculated with respect to the final steady state of the deployment simulation. For instance, for the deployment of the cm $\to$ cm $\to$ pm pattern shown in Fig.~\ref{fig:deployment}, the SCR value is given by the ratio of the area of the final pm state (bottom right) to the area of the initial cm state (top left).

In Table~\ref{tab:patterns}, we consider the 17 monotile kirigami patterns presented in Fig.~\ref{fig:monotile_proof} and study the SCR and PCR of them. Note that because of their periodicity, it suffices to study the SCR and PCR for a basic unit cell of these patterns, and increasing the pattern resolution will not lead to a significant change in either the SCR or PCR. Here, we utilize the deployment simulator~\cite{liu2022quasicrystal} to deploy each pattern and compute both the SCR and PCR in real time using the deployed configuration at each time point. The deployment continues until both SCR and PCR differ from their corresponding averages over the last 10 iterations by less than $0.001$, and after that we record the final SCR and PCR values. It can be observed that the p6m$ \rightarrow $ p6m pattern has the highest SCR and PCR, while the p3m1 $\rightarrow $ pmm pattern has the lowest SCR and PCR. Moreover, it is noteworthy that while many of the patterns considered have the same connectivity, their SCR and PCR can be highly different. This shows that the designs of both the tile geometry and the tile connectivity are essential for controlling the size change effect of the deployable monotile kirigami structures. Besides the final SCR and PCR, we can further analyze the changes of SCR and PCR throughout the entire deployment process as shown in Fig.~\ref{fig:scr_pcr_change}. Here, the entire simulation time interval is normalized to $[0,1]$ for a consistent comparison between different structures. It can be observed that neglecting some small fluctuations caused by the numerical computation of the simulator, all patterns give a similar trend for both SCR and PCR throughout the entire deployment. In particular, the SCR consistently exceeds the PCR across all time points and patterns, and both of them always exhibit a more rapid increase near the start of deployment before eventually stabilizing.

\begin{figure}[t]
    \centering
    \includegraphics[width=\linewidth]{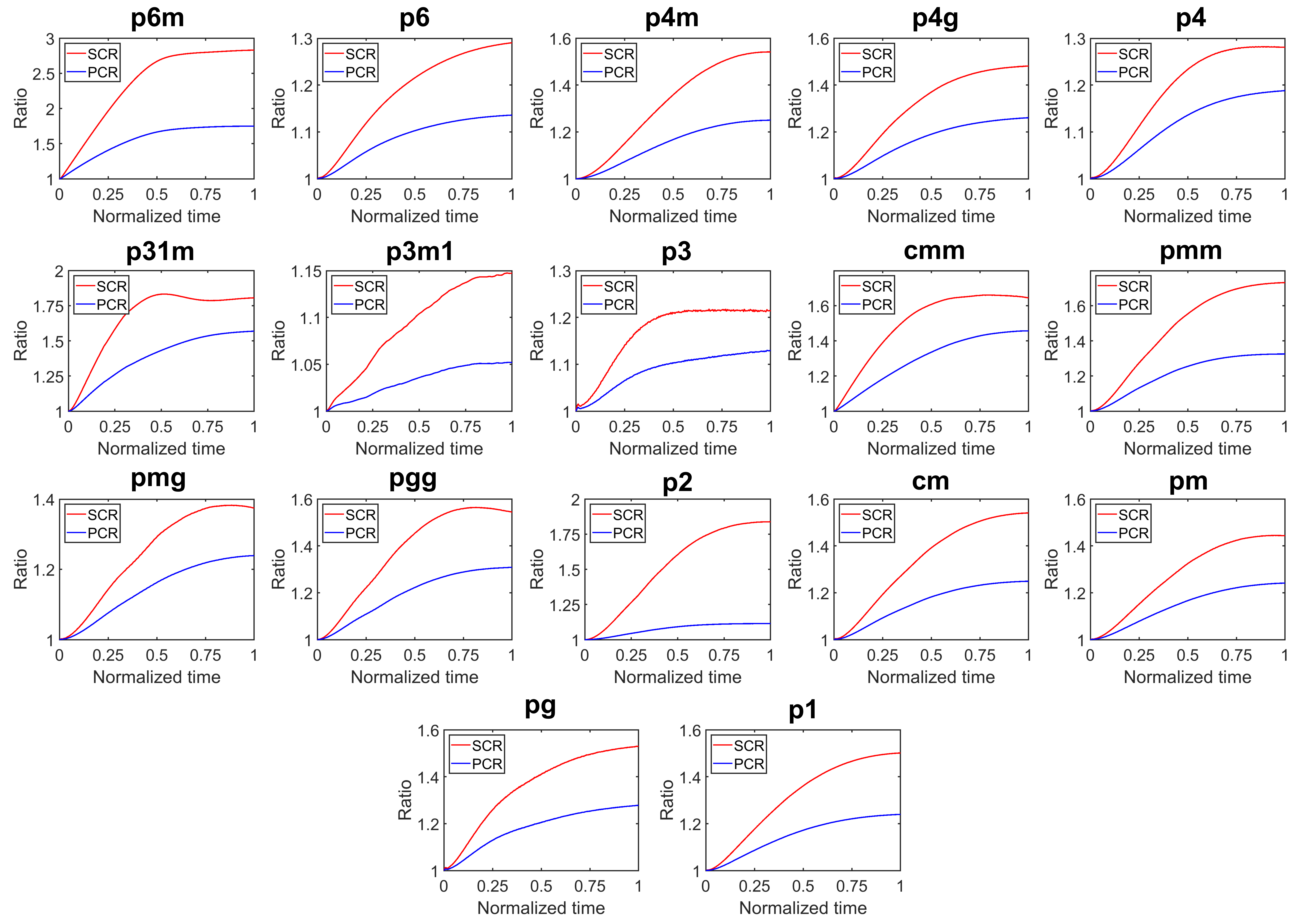}
    \caption{\textbf{The size change ratio (SCR) and perimeter change ratio (PCR) throughout the continuous deployment process.} For each of the 17 monotile kirigami patterns presented in Fig.~\ref{fig:monotile_proof}, we record the values of SCR and PCR of the resulting pattern at each time point throughout the deployment simulations. The deployment time is normalized to $[0,1]$ for all patterns.}
    \label{fig:scr_pcr_change}
\end{figure}

After studying the periodic monotile kirigami patterns, we move on to the analysis of the aperiodic monotile kirigami patterns (see Table~\ref{tab:resolutions}). In particular, here we consider both the quasicrystal monotile kirigami patterns (Penrose, Ammann--Beenker, Stampfli) presented in Fig.~\ref{fig:quasicrystal} and two representative examples of polykite monotile kirigami patterns formed by ``Hat'' and ``Turtle''. In contrast to the previous analysis of the wallpaper group monotile kirigami patterns, no periodicity is present in the aperiodic monotile kirigami patterns here. Hence, it is also natural to consider how an increase in the pattern resolution will lead to a change in the size change. As shown in the results, increasing the pattern resolution of the quasicrystal monotile kirigami patterns may lead to a significant change in the SCR and PCR, while the SCR and PCR of the polykite monotile kirigami patterns are relatively stable even if we increase the resolution.

\begin{table}[t]
    \centering
    \begin{tabular}{c|c|c|c}
        \textbf{Pattern} & \textbf{Number of tiles} & \textbf{Size change ratio} & \textbf{Perimeter change ratio} \\ \hline 
        Penrose &20& 1.2521&1.1190\\
        Penrose &35& 1.8304&1.3529\\
        Penrose &70& 1.3445&1.1550\\
        Ammann--Beenker &40&1.1924 &1.0920\\
        Ammann--Beenker &64& 2.9941&1.7303\\
        Stampfli & 48 & 1.5701 & 1.2530 \\
        Stampfli & 84 & 1.3943 & 1.2331 \\
        Hat & 6 & 1.2283 & 1.1526\\
        Hat & 36 & 1.1969 &1.1365\\
        Turtle & 6 & 1.1560 & 1.1251\\
        Turtle & 36 & 1.1954 & 1.1161\\
    \end{tabular}
    \caption{Analysis of aperiodic monotile kirigami patterns, including Penrose, Ammann--Beenker, Stampfli, Hat, and Turtle, with different pattern resolutions.}
    \label{tab:resolutions}
\end{table}

\subsection{More analyses on the hat patterns and variants}

\begin{table}[t!]
    \centering
    \begin{tabular}{c|c|c}
     \textbf{Tile($a,b$)} & \textbf{Size change ratio} & \textbf{Perimeter change ratio}\\ \hline
     Tile(0,2.73)&1.4098  &1.2313  \\
     Tile(0.37,2.36)&1.3257&1.2042\\
     Tile(0.65,2.08)&1.2712&1.1771\\
     Tile(1.00,1.73) (Hat)&1.2283&1.1526\\
     Tile(1.23,1.51)&1.1785&1.1265\\
     Tile(1.51,1.23)&1.1726&1.1271\\
     Tile(1.73,1.00) (Turtle) &1.1560&1.1251\\
     Tile(2.08,0.65)&1.1405&1.1203\\
     Tile(2.36,0.37)&1.1359&1.1175\\
     Tile(2.73,0)&1.1342 &1.1135  \\
    \end{tabular}
    \caption{Analysis of the size change ratio and perimeter change ratio of aperiodic polykite monotile kirigami patterns constructed using polykite monotiles in the form of Tile$(a,b)$ with different tile geometry parameters $(a,b)$.}
    \label{tab:differentab}
\end{table}

Recall that both the ``Hat'' and ``Turtle'' tiles belong to an infinite family of polykite monotiles in the form of Tile$(a,b)$, where $a$ and $b$ are the relative lengths of two consecutive sides of the tile~\cite{smith2024aperiodic}. As shown in Table~\ref{tab:resolutions}, the ``Hat'' and ``Turtle'' monotile kirigami structures yield different SCR and PCR. It is natural to ask how the geometric properties of the polykite monotile kirigami structures change with the parameters $a, b$. Therefore, here we further consider varying the ratio between $a, b$ as in Fig.~\ref{fig:polykite}(e), covering $(a,b) = (0.37,2.36)$, $(0.65,2.08)$, $(1.00,1.73)$ (which corresponds to ``Hat''), $(1.23,1.51)$, $(1.51,1.23)$, $(1.73,1.00)$ (which corresponds to ``Turtle''), $(2.08,0.65)$, and $(2.36,0.37)$. The two extreme cases (0, 2.73) and (2.73, 0), which are the bounds for the family with one of the edges being degenerated, are also covered to study the limits of the observed behaviors.

In Table~\ref{tab:differentab}, we report the SCR and PCR of the resulting polykite monotile kirigami structures. It can be observed that in general, as the ratio $a/b$ increases, both the SCR and the PCR decrease. To further quantify this relationship, in Fig.~\ref{fig:scr_pcr_graph} we plot the logged ratios $\log\text{(SCR)}$ and $\log\text{(PCR)}$ against $\log(b/a)$ (with the two cases with degenerated edges excluded as $\log(b/a)$ is not well-defined), from which we can see that both $\log\text{(SCR)}$ and $\log\text{(PCR)}$ increase as $\log(b/a)$ increases. By calculating the slope of the best-fit least-squares straight lines, we have
\begin{equation}
    \text{SCR}(a,b)\sim (b/a)^{0.0436}
\end{equation}
and
\begin{equation}
    \text{PCR}(a,b) \sim (b/a)^{0.0204}.
\end{equation}
This suggests a simple way to achieve any desired SCR and PCR for an aperiodic polykite monotile kirigami pattern by suitably controlling the ratio between the polykite tile side lengths $a$ and $b$.

\begin{figure}[t!]
    \centering
    \includegraphics[width=0.8\linewidth]{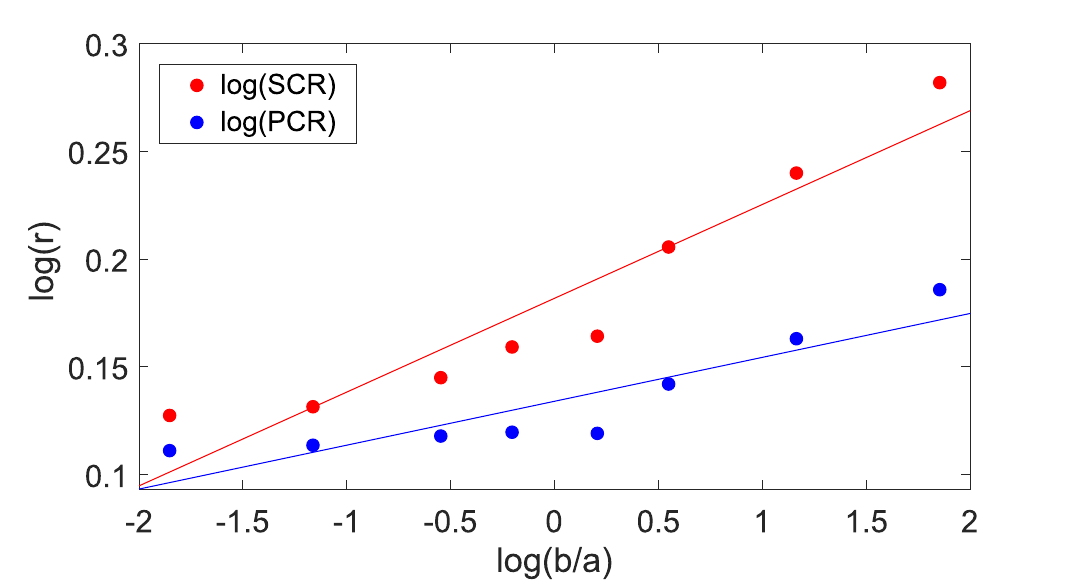}
    \caption{\textbf{Analyzing the relationship between the tile geometry and the size change of aperiodic polykite monotile kirigami patterns.} Here, we plot $\log(r)$ (where $r$ is the size change ratio (SCR) or perimeter change ratio (PCR)) against $\log(b/a)$ for different choices of tile geometry parameters $(a,b)$. The best-fit least-squares straight lines are also displayed.}
    \label{fig:scr_pcr_graph}
\end{figure}

\subsection{Analysis on Poisson's ratio}
Besides the size and shape change of the monotile kirigami patterns, we can further extend our computational analysis for investigating their Poisson’s ratio. Specifically, for each input pattern, we select two tiles at the top and two at the bottom, and use their centers as the four loading points. We then apply vertical spring forces to deploy the pattern and compute the Poisson's ratio 
\begin{equation}
    \nu = - \frac{\epsilon_x}{\epsilon_y} = -\frac{(W-W_0)/W_0}{(H-H_0)/H_0},
\end{equation}
where $W_0$, $H_0$ are the width and height of the initial structure, and $W, H$ are the width and height of the deployed structure. Here, $W_0, H_0, W, H$ are all computed using the axis-aligned bounding box of the corresponding structure. 

Table~\ref{tab:poisson} records the Poisson's ratio of all periodic monotile kirigami structures (left column) and aperiodic monotile kirigami structures (right column) considered in this work. It can be observed that most of the periodic monotile kirigami structures are auxetic, while the Poisson's ratios of aperiodic monotile kirigami structures vary across different patterns and resolutions.

\begin{table}[t]
    \centering
    \begin{tabular}{c|c||c|c}
        \textbf{Pattern} & \textbf{Poisson's ratio} & \textbf{Pattern} & \textbf{Poisson's ratio}  \\ \hline
         p6m $\to$ p31m $\to$ p6m & $-0.7489$ & Penrose (20) &  $0.3271$\\
         p6 $\to$ p6 $\to$ p6 & $-0.9765$ & Penrose (35)  & $-0.1468$ \\
         p4m $\to$ p4g $\to$ p4m  & $-0.9995$ & Penrose (70) & $-0.1381$ \\
         p4g $\to$ pgg $\to$ cmm & $-0.9571$ & Ammann--Beenker (40) & $0.4931$ \\
         p4 $\to$ p4 $\to$ p4 & $-0.8095$ & Ammann--Beenker (64) & $0.1266$ \\
        p31m $\to$ p3 $\to$ p3 & $-0.1820$ & Stampfli (48) & $0.7643$ \\
        p3m1 $\to$ p1 $\to$ pmm & $0.2174$ & Stampfli (84) & $0.3918$ \\
        p3 $\to$ p3 $\to$ p6  & $0.0726$ & Tile(0,2.73) & $-3.6146$ \\
        cmm $\to$ p2 $\to$ p2 & $0.1546$ & Tile(0.37,2.36) & $-1.5556$ \\
         pmm $\to$ pmm $\to$ pmm & $-0.3336$& Tile(0.65,2.08)  & $-0.5620$ \\
         pmg $\to$ pgg $\to$ p4g & $-1.3947$ & Tile(1.00,1.73) (Hat) & $-1.3531$ \\
         pgg $\to$ pgg $\to$ cmm & $-0.6981$ & Tile(1.23,1.51) & $-0.9801$ \\
         p2 $\to$ p2 $\to$ p4 & $0.1344$ & Tile(1.51,1.23) & $3.2920$ \\
          cm $\to$ cm $\to$ pm & $-0.9993$ & Tile(1.73,1.00) (Turtle) & $3.8695$ \\
         pm $\to$ cm $\to$ cm & $-1.2490$ & Tile(2.08,0.65)  & $2.6289$ \\
         pg $\to$ pg $\to$ cm  & $-0.8121$ & Tile(2.36,0.37) & $2.4008$ \\
         
         p1 $\to$ p1 $\to$ p1 & $-1.2494$ & Tile(2.73,0) & $-1.4113$ \\

    \end{tabular}
    \caption{Poisson's ratio of different monotile kirigami patterns obtained from our numerical simulations. The left part corresponds to the 17 periodic monotile kirigami patterns considered in Fig.~\ref{fig:monotile_proof}. The right part corresponds to the aperiodic monotile kirigami patterns considered in this work, including the quasicrystal tilings with different tile resolutions as shown in Fig.~\ref{fig:quasicrystal} and the polykite monotile kirigami patterns in Fig.~\ref{fig:polykite}.}
    \label{tab:poisson}
\end{table}

%%%%%%%%%%%%%%%%%%%%%%%%%%%%%%%%%%%%%%%%%

\section{Discussion and Future Work} \label{sect:conclusion}

In this work, we have explored the design of periodic and aperiodic monotile kirigami patterns. In particular, we have proved the existence of deployable kirigami structures based on monotile patterns, and have further performed a systematic analysis of their geometrical properties using both theory and computation. 

As shown in our analysis, monotile kirigami patterns exhibit highly surprising properties. Specifically, even with the use of only one type of tile, it is possible to achieve periodic and aperiodic deployable structures with a wide range of symmetry properties. The possibility of achieving gain, loss, and preservation of rotational, reflectional, and glide reflectional symmetries makes monotile kirigami patterns highly desirable for different practical applications that require specific shape-shifting effects. The analysis on the size change of different periodic and aperiodic monotile kirigami patterns also provides a simple guideline for users to select suitable tile geometries and their connections, or further modify them, to achieve different target outcomes. More broadly, the monotile kirigami patterns can naturally facilitate the use of kirigami-inspired deployable structures in various science and engineering applications, as only one single shape is required for the tile cutting or moulding process during mass manufacturing.

Moreover, our study of monotile kirigami can serve as a solid foundation for further exploration of the design of more general shape-morphing structures. For instance, while we have only considered the design of planar kirigami structures using monotile patterns in this work, it has recently been shown that the 2D symmetry groups can also be used for the design of kinematic-based origami metastructures~\cite{liu2025reconfigurable}. Therefore, it is natural to explore the possibility of extending our study to origami design. Another possible direction is to consider the three-dimensional analog of the monotile patterns for the design of three-dimensional metamaterials~\cite{choi2020control,maurizi2025designing} and analyze their geometric and mechanical properties for a wider range of applications. 

\bibliographystyle{ieeetr}
\bibliography{reference}

\begin{thebibliography}{10}

\bibitem{zhai2021mechanical}
Z.~Zhai, L.~Wu, and H.~Jiang, ``Mechanical metamaterials based on origami and kirigami,'' {\em Applied Physics Reviews}, vol.~8, no.~4, 2021.

\bibitem{jiao2023mechanical}
P.~Jiao, J.~Mueller, J.~R. Raney, X.~Zheng, and A.~H. Alavi, ``Mechanical metamaterials and beyond,'' {\em Nature Communications}, vol.~14, no.~1, p.~6004, 2023.

\bibitem{jin2024engineering}
L.~Jin and S.~Yang, ``Engineering kirigami frameworks toward real-world applications,'' {\em Advanced Materials}, vol.~36, no.~9, p.~2308560, 2024.

\bibitem{choi2024computational}
G.~P.~T. Choi, ``Computational design of art-inspired metamaterials,'' {\em Nature Computational Science}, vol.~4, no.~8, pp.~549--552, 2024.

\bibitem{dudek2025shape}
K.~K. Dudek, M.~Kadic, C.~Coulais, and K.~Bertoldi, ``Shape-morphing metamaterials,'' {\em Nature Reviews Materials}, vol.~10, pp.~783--798, 2025.

\bibitem{blees2015graphene}
M.~K. Blees, A.~W. Barnard, P.~A. Rose, S.~P. Roberts, K.~L. McGill, P.~Y. Huang, A.~R. Ruyack, J.~W. Kevek, B.~Kobrin, D.~A. Muller, {\em et~al.}, ``Graphene kirigami,'' {\em Nature}, vol.~524, no.~7564, p.~204, 2015.

\bibitem{rafsanjani2018kirigami}
A.~Rafsanjani, Y.~Zhang, B.~Liu, S.~M. Rubinstein, and K.~Bertoldi, ``Kirigami skins make a simple soft actuator crawl,'' {\em Science Robotics}, vol.~3, no.~15, p.~eaar7555, 2018.

\bibitem{wang2023physics}
L.~Wang, Y.~Chang, S.~Wu, R.~R. Zhao, and W.~Chen, ``Physics-aware differentiable design of magnetically actuated kirigami for shape morphing,'' {\em Nature Communications}, vol.~14, no.~1, p.~8516, 2023.

\bibitem{yang2023new}
Y.~Yang, A.~Vallecchi, E.~Shamonina, C.~J. Stevens, and Z.~You, ``A new class of transformable kirigami metamaterials for reconfigurable electromagnetic systems,'' {\em Scientific Reports}, vol.~13, no.~1, p.~1219, 2023.

\bibitem{grunbaum1986tilings}
B.~Gr{\"u}nbaum and G.~C. Shephard, {\em Tilings and patterns}.
\newblock WH Freeman \& Co., 1986.

\bibitem{grima2006auxetic}
J.~N. Grima and K.~E. Evans, ``Auxetic behavior from rotating triangles,'' {\em Journal of Materials Science}, vol.~41, no.~10, pp.~3193--3196, 2006.

\bibitem{grima2000auxetic}
J.~N. Grima and K.~E. Evans, ``Auxetic behavior from rotating squares,'' {\em Journal of Materials Science Letters}, vol.~19, no.~17, pp.~1563--1565, 2000.

\bibitem{grima2004negative}
J.~N. Grima, A.~Alderson, and K.~E. Evans, ``Negative poisson’s ratios from rotating rectangles,'' {\em Computational Methods in Science and Technology}, vol.~10, no.~2, pp.~137--145, 2004.

\bibitem{attard2008auxetic}
D.~Attard and J.~N. Grima, ``Auxetic behaviour from rotating rhombi,'' {\em Physica Status Solidi B}, vol.~245, no.~11, pp.~2395--2404, 2008.

\bibitem{rafsanjani2016bistable}
A.~Rafsanjani and D.~Pasini, ``Bistable auxetic mechanical metamaterials inspired by ancient geometric motifs,'' {\em Extreme Mechanics Letters}, vol.~9, pp.~291--296, 2016.

\bibitem{grima2011auxetic}
J.~N. Grima, E.~Manicaro, and D.~Attard, ``Auxetic behaviour from connected different-sized squares and rectangles,'' {\em Proceedings of the Royal Society A}, vol.~467, no.~2126, pp.~439--458, 2011.

\bibitem{stavric2019geometrical}
M.~Stavric and A.~Wiltsche, ``Geometrical elaboration of auxetic structures,'' {\em Nexus Network Journal}, vol.~21, no.~1, pp.~79--90, 2019.

\bibitem{liu2021wallpaper}
L.~Liu, G.~P.~T. Choi, and L.~Mahadevan, ``Wallpaper group kirigami,'' {\em Proceedings of the Royal Society A}, vol.~477, no.~2252, p.~20210161, 2021.

\bibitem{liu2024auxetic}
X.~Liu, L.~Lu, L.~Cao, O.~Deussen, and C.~Tu, ``Auxetic dihedral {E}scher tessellations,'' {\em Graphical Models}, vol.~133, p.~101215, 2024.

\bibitem{liu2022quasicrystal}
L.~Liu, G.~P.~T. Choi, and L.~Mahadevan, ``Quasicrystal kirigami,'' {\em Physical Review Research}, vol.~4, no.~3, p.~033114, 2022.

\bibitem{choi2019programming}
G.~P.~T. Choi, L.~H. Dudte, and L.~Mahadevan, ``Programming shape using kirigami tessellations,'' {\em Nature Materials}, vol.~18, no.~9, pp.~999--1004, 2019.

\bibitem{choi2021compact}
G.~P.~T. Choi, L.~H. Dudte, and L.~Mahadevan, ``Compact reconfigurable kirigami,'' {\em Physical Review Research}, vol.~3, no.~4, p.~043030, 2021.

\bibitem{dudte2023additive}
L.~H. Dudte, G.~P.~T. Choi, K.~P. Becker, and L.~Mahadevan, ``An additive framework for kirigami design,'' {\em Nature Computational Science}, vol.~3, no.~5, pp.~443--454, 2023.

\bibitem{lubbers2019excess}
L.~A. Lubbers and M.~van Hecke, ``Excess floppy modes and multibranched mechanisms in metamaterials with symmetries,'' {\em Physical Review E}, vol.~100, no.~2, p.~021001, 2019.

\bibitem{an2020programmable}
N.~An, A.~G. Domel, J.~Zhou, A.~Rafsanjani, and K.~Bertoldi, ``Programmable hierarchical kirigami,'' {\em Advanced Functional Materials}, vol.~30, no.~6, p.~1906711, 2020.

\bibitem{chen2020deterministic}
S.~Chen, G.~P.~T. Choi, and L.~Mahadevan, ``Deterministic and stochastic control of kirigami topology,'' {\em Proceedings of the National Academy of Sciences}, vol.~117, no.~9, pp.~4511--4517, 2020.

\bibitem{choi2023explosive}
G.~P.~T. Choi, L.~Liu, and L.~Mahadevan, ``Explosive rigidity percolation in kirigami,'' {\em Proceedings of the Royal Society A}, vol.~479, no.~2271, p.~20220798, 2023.

\bibitem{grima2010perforated}
J.~N. Grima and R.~Gatt, ``Perforated sheets exhibiting negative {P}oisson's ratios,'' {\em Advanced Engineering Materials}, vol.~12, no.~6, pp.~460--464, 2010.

\bibitem{gatt2015hierarchical}
R.~Gatt, L.~Mizzi, J.~I. Azzopardi, K.~M. Azzopardi, D.~Attard, A.~Casha, J.~Briffa, and J.~N. Grima, ``Hierarchical auxetic mechanical metamaterials,'' {\em Scientific Reports}, vol.~5, no.~1, p.~8395, 2015.

\bibitem{taylor2014low}
M.~J. Taylor, L.~Francesconi, M.~Gerend{\'a}s, A.~Shanian, C.~Carson, and K.~Bertoldi, ``Low porosity metallic periodic structures with negative {P}oisson's ratio,'' {\em Advanced Materials}, vol.~26, no.~15, 2014.

\bibitem{lim2026auxetic}
T.-C. Lim, ``Auxetic metamaterial model made from rotating rectangles, hexagons, and triangles,'' {\em Physica Status Solidi (b)}, vol.~263, no.~5, p.~e202400343, 2026.

\bibitem{lim2024metamaterial}
T.-C. Lim, ``A metamaterial with sign-switching and discontinuous {P}oisson’s ratio,'' {\em International Journal of Solids and Structures}, vol.~305, p.~113095, 2024.

\bibitem{mizzi2021chiralisation}
L.~Mizzi and A.~Spaggiari, ``Chiralisation of euclidean polygonal tessellations for the design of new auxetic metamaterials,'' {\em Mechanics of Materials}, vol.~153, p.~103698, 2021.

\bibitem{tarnai2022equiauxetic}
T.~Tarnai, P.~W. Fowler, S.~D. Guest, and F.~Kov{\'a}cs, ``Equiauxetic hinged archimedean tilings,'' {\em Symmetry}, vol.~14, no.~2, p.~232, 2022.

\bibitem{mizzi2023design}
L.~Mizzi, L.~Grasselli, A.~Spaggiari, R.~Gatt, P.-S. Farrugia, and J.~N. Grima, ``Design of isotropic {2D} chiral metamaterials based on monohedral pentagonal tessellations,'' {\em Thin-Walled Structures}, vol.~187, p.~110739, 2023.

\bibitem{smith2024aperiodic}
D.~Smith, J.~S. Myers, C.~S. Kaplan, and C.~Goodman-Strauss, ``An aperiodic monotile,'' {\em Combinatorial Theory}, vol.~4, no.~1, p.~6, 2024.

\bibitem{smith2024chiral}
D.~Smith, J.~S. Myers, C.~S. Kaplan, and C.~Goodman-Strauss, ``A chiral aperiodic monotile,'' {\em Combinatorial Theory}, vol.~4, no.~2, p.~13, 2024.

\bibitem{clarke2023isotropic}
D.~J. Clarke, F.~Carter, I.~Jowers, and R.~J. Moat, ``An isotropic zero poisson's ratio metamaterial based on the aperiodic ‘hat’monotile,'' {\em Applied Materials Today}, vol.~35, p.~101959, 2023.

\bibitem{jung2024aperiodicity}
J.~Jung, A.~Chen, and G.~X. Gu, ``Aperiodicity is all you need: Aperiodic monotiles for high-performance composites,'' {\em Materials Today}, vol.~73, pp.~1--8, 2024.

\bibitem{naji2024effective}
M.~M. Naji and R.~K.~A. Al-Rub, ``Effective elastic properties of novel aperiodic monotile-based lattice metamaterials,'' {\em Materials \& Design}, vol.~244, p.~113102, 2024.

\bibitem{schirmann2024physical}
J.~Schirmann, S.~Franca, F.~Flicker, and A.~G. Grushin, ``Physical properties of an aperiodic monotile with graphene-like features, chirality, and zero modes,'' {\em Physical Review Letters}, vol.~132, no.~8, p.~086402, 2024.

\bibitem{liu2025reconfigurable}
W.~Liu, Q.~Ren, Y.~Wang, Z.~Zhang, X.~Wang, Y.~Liang, H.~Huang, J.~Ma, H.~Jiang, and Y.~Chen, ``Reconfigurable modular origami for tunable {2D} symmetry groups,'' {\em Science Advances}, vol.~11, no.~46, p.~eady3812, 2025.

\bibitem{choi2020control}
G.~P.~T. Choi, S.~Chen, and L.~Mahadevan, ``Control of connectivity and rigidity in prismatic assemblies,'' {\em Proceedings of the Royal Society A}, vol.~476, no.~2244, p.~20200485, 2020.

\bibitem{maurizi2025designing}
M.~Maurizi, D.~Xu, Y.-T. Wang, D.~Yao, D.~Hahn, M.~Oudich, A.~Satpati, M.~Bauchy, W.~Wang, Y.~Sun, {\em et~al.}, ``Designing metamaterials with programmable nonlinear responses and geometric constraints in graph space,'' {\em Nature Machine Intelligence}, vol.~7, no.~7, pp.~1023--1036, 2025.

\end{thebibliography}

\end{document}